\documentclass[11pt,a4paper,english]{amsart}
\usepackage{babel}
\usepackage[latin1]{inputenc}
\usepackage{amsmath}
\usepackage{amsthm}
\usepackage{amsfonts}
\usepackage{indentfirst}
\usepackage{graphicx}
\usepackage{amssymb}
\usepackage[mathscr]{eucal}
\usepackage{tikz}
\usepackage{caption}
\usepackage{amsthm}
\usepackage{amscd}
\newtheorem{teo}{Theorem}%[section]

\newtheorem{pro}[teo]{Proposition}%[section]
\newtheorem{lem}[teo]{Lemma}%[section]

\theoremstyle{definition}
\newtheorem{rem}[teo]{Remark}%[section]
\newtheorem{de}[teo]{Definition}%[section]
\newtheorem{exa}{Example}%[section]
\title{Classical and Quantum Evaluation Codes at the Trace Roots}
\author{Carlos Galindo, Fernando Hernando and Diego Ruano}
\curraddr{\texttt{Carlos Galindo and Fernando Hernando:} Instituto
Universitario de Matem\'aticas y Aplicaciones de Castell\'on and
Departamento de Matem\'aticas, Universitat Jaume I, Campus de Riu
Sec. 12071 Castell\'{o} (Spain)\\
\texttt{Diego Ruano:} Department of Mathematical Sciences, Aalborg University, Skjernvej 4A, 9220 Aalborg East (Denmark).
}
\email{{\rm Galindo:} galindo@uji.es; {\rm Hernando:} carrillf@uji.es; {\rm Ruano:} diego@math.aau.dk}
\date{}
\thanks{Supported by the Spanish Ministry of Economy/FEDER: grants MTM2015-65764-C3-2-P and MTM2015-69138-REDT, the University Jaume I: grant PB1-1B2015-02 and the Danish Council for Independent Research, grant DFF-4002-00367.}

\begin{document}

\begin{abstract}
We introduce a new class of evaluation linear codes by evaluating polynomials at the roots of a suitable trace function. We give conditions for self-orthogonality of these codes and their subfield-subcodes with respect to the Hermitian inner product. They allow us to construct stabilizer quantum  codes over several finite fields which substantially improve the codes in the literature and that are records at \cite{codet} for the binary case. Moreover, we obtain several classical linear codes over the field $\mathbb{F}_4$ which are records at \cite{codet}.
\end{abstract}

\maketitle

\section{Introduction}
A stabilizer (quantum) code $\mathcal{C} \neq \{0\}$ is the common eigenspace of a commutative subgroup of the error group generated by a nice error basis on the space $\mathbb{C}^{q^n}$, where $\mathbb{C}$ denotes the complex numbers, $q$ is a positive power of a prime number and $n$ is a positive integer \cite{kkk}.  The code $\mathcal{C}$ has minimum distance $d$ as long as errors with weight less than $d$ can be detected or have no effect on $\mathcal{C}$ but some error with weight $d$ cannot be detected. Furthermore, if $\mathcal{C}$ has dimension $q^k$ as a $\mathbb{C}$-vector space, then we say that the  code $\mathcal{C}$ has parameters $[[n,k,d]]_q$.

The importance of quantum computation is beyond doubt after \cite{22RBC}, polynomial time algorithms for prime factorization and discrete logarithms on quantum computers have been given. Quantum error-correcting codes are essential for this type of computation since they protect quantum information from decoherence and quantum noise. Quantum codes were first introduced for the binary case, some references are \cite{7kkk, 8kkk, 18kkk, 19kkk, 20kkk, 38kkk, 45kkk},  and, subsequently, for the general case (see for instance \cite{AK, BE, 35kkk, opt, lag3, 71kkk}). The interest on the general case continues to grow, especially after the realization that these codes are useful for fault-tolerant computation.

Stabilizer codes can be constructed from self-orthogonal classical linear codes:

\begin{teo}
\label{el1}
\cite{kkk,Akk}
Let $C$ be a linear $[n,k,d]$ error-correcting  code over the field $\mathbb{F}_{q^2}$ such that $C^{\perp_h }\subseteq C$. Then, there exists an $[[n, 2k -n, \geq d]]_q$ stabilizer code.
\end{teo}

The symbol $\perp_h$ means dual with respect to Hermitian inner product. An analogous result also holds for Euclidean duality when $C$ is defined over $\mathbb{F}_q$, which gives rise to quantum codes obtained from the CSS construction \cite{20kkk,95kkk}. In this paper, most of our codes will be derived from Theorem \ref{el1}. Although quantum codes were introducecd recently, the literature on this topic is very large. Most papers have addressed the study of quantum MDS, LDCP and BCH codes \cite{Sarvepalli, edel, Akk, lag3, lag1, jin,  f2, lag2, refer}.

In this paper, we introduce a new family of classical linear codes, they are evaluation codes of polynomials in one variable at the set of zeros of a suitable trace map (see Definition \ref{TREC}). The algebraic structure of the set of zeros of the trace map allows us to consider suitable subfield-subcodes, providing a new family of subfield-subcodes different from BCH codes, extended BCH codes or $J$-affine variety codes \cite{galindo-hernando, gal-her-rua, QINP, QINP2}. For designing our codes, we will use {\it consecutive} cyclotomic cosets,  the size and number of these cosets will determine a designed distance and a lower bound for the dimension.

Although we are mainly interested in quantum codes, this new family of classical linear codes allows us to obtain 50 linear code records at \cite{codet} (see Example \ref{ex:uno} in Section \ref{sect3}). We construct linear codes with parameters $[128,85,16]_4$, $[128,79,20]_4$ and $[128,75,22]_4$ improving those with the same length and dimension in \cite{codet}. The remaining records are obtained by shortening the above three codes.

In Theorem \ref{25}, we study the dimension and minimum distance of the subfield-subcodes of this new family of codes and in Theorem \ref{26}, we give conditions for their self-orthogonality with respect to Hermitian inner product. In sum, from linear codes  over $\mathbb{F}_{p^{2r}}$, $p$ a prime number, we get linear codes over $\mathbb{F}_{p^{2s}}$, $s$ being a positive integer that divides $r$,  which give quantum codes over $\mathbb{F}_{p^{s}}$ with good parameters, improving those in the literature.

Apart from the introduction, this paper contains four sections. The definition of our codes and conditions for their self-intersection with respect to Hermitian inner product are given in Section \ref{sect1}. Fundamental results on subfield-codes are presented in Section \ref{sect:subfield}, we will follow the approach in \cite{galindo-hernando, gal-her-rua, QINP, QINP2} for $J$-affine variety codes.  Section \ref{sect2} is the core of the paper, where we consider stabilizer codes from the classical codes defined in the previous section. We also consider codes defined by evaluating at the non-roots of the trace function as well, we will refer to these cades as complementary codes. Finally, Section \ref{sect3} is devoted to provide good examples of our codes.  Apart from the above mentioned classical linear code records, we also give several examples of binary stabilizer quantum codes improving the records at \cite{codet}. In addition, we give tables containing stabilizer codes over $\mathbb{F}_4, \mathbb{F}_5$ and $ \mathbb{F}_7$. For comparing our codes, we consider the codes in \cite{lag3} and show that our codes largely improve them. We also provide new codes with a length that did not exist in the literature and, almost all of them, exceed the quantum Gilbert-Varshamov bounds \cite{mat, feng, kkk}.

\section{Evaluation Codes at the Trace Roots}
\label{sect1}

We devote this section to introduce a new class of evaluation linear codes and study their behavior under Hermitian duality. We are mainly interested in quantum codes although it is worthwhile to mention that their subfield-subcodes provide good classical codes as well. Their subfield-subcodes will be treated in Section \ref{sect2}.

Throughout this paper, let $p$ be a prime number and $r$ and $s$ positive integers such that $s|r$. Set $r=s\cdot n$ and $q=p^s$. Our procedure to obtain stabilizer quantum codes over $\mathbb{F}_q = \mathbb{F}_{p^s}$, using Theorem \ref{el1}, consists of considering subfield-subcodes over $\mathbb{F}_{p^{2s}}$ of classical linear codes over over $\mathbb{F}_{p^{2r}}$.

The {\it trace polynomial} over $\mathbb{F}_{p^{2r}} = \mathbb{F}_{q^{2n}}$ with respect to $\mathbb{F}_{q}$ is defined as
\[
\mathrm{tr}_{2r}^s (X) = X+ X^q+ X^{q^2}+ \cdots + X^{q^{2n-1}},
\]
whose attached polynomial function ({\it trace map}) will be denoted by $\mathrm{tr}_{2r}^s: \mathbb{F}_{q^{2n}} \rightarrow \mathbb{F}_{q}$.

It is well-known that the trace map is a linear transformation over $\mathbb{F}_{q}$ and any linear transformation $\mathbb{F}_{q^{2n}} \rightarrow \mathbb{F}_{q}$ is defined by $x \mapsto \mathrm{tr}_{2r}^s (\beta x)$, for some $\beta \in \mathbb{F}_{q^{2n}} $. Another interesting property of the trace map is that
\[
\mathrm{card} \left\{ \alpha \in \mathbb{F}_{q^{2n}} | \mathrm{tr}_{2r}^s (\alpha) = a \right\}
\]
equals $q^{2n-1}$  for all $a \in \mathbb{F}_q$, and therefore, when $\alpha$ runs over $\mathbb{F}_{q^{2n}}$, one has that $\mathrm{tr}_{2r}^s (\alpha)$ takes each value of $\mathbb{F}_q$ exactly $q^{2n-1}$ times. This fact gives rise to the decomposition
\[
\mathrm{tr}_{2r}^s (X) - a = \prod_{\alpha \in \mathbb{F}_{q^{2n}}, \mathrm{tr}_{2r}^s (\alpha) = a} \left(X- \alpha \right)
\]
and, as a consequence,
\[
X^{q^{2n}} - X = \prod_{a \in \mathbb{F}_q} \left( \mathrm{tr}_{2r}^s (X) - a \right).
\]

Consider now the ideal of the polynomial ring $\mathbb{F}_{q^{2n}}[X]$ generated by $\mathrm{tr}_{2r}^s (X)$, which, by the previous discussion, can also be regarded as the ideal generated by both polynomials $X^{q^{2n}} - X$ and $\mathrm{tr}_{2r}^s (X)$. Consider also
\[
Z= \left\{ \alpha \in \mathbb{F}_{q^{2n}} | \mathrm{tr}_{2r}^s (\alpha) =0 \right\} =\left\{ \alpha_1, \alpha_2, \ldots, \alpha_N \right\},
\]
where $N= q^{2n-1}$.

Next, we define the evaluation map that supports our codes:
\[
\mathrm{ev}_{\mathrm{tr}_{2r}^s}: \mathbb{F}_{q^{2n}}[X] /\langle \mathrm{tr}_{2r}^s (X) \rangle \longrightarrow \mathbb{F}_{q^{2n}}^N, \;\; \mathrm{ev}_{\mathrm{tr}_{2r}^s} (f) = ( f(\alpha_1), f(\alpha_2), \ldots, f(\alpha_N)),
\]
where $f$ denotes both the class in $\mathbb{F}_{q^{2n}}[X] /\langle \mathrm{tr}_{2r}^s (X) \rangle$ and a polynomial in $\mathbb{F}_{q^{2n}}[X] $ representing that class. Notice that we have proved that the map $\mathrm{ev}_{\mathrm{tr}_{2r}^s}$ is well-defined.

Our codes will take advantage from the existing relations in the ring $\mathbb{F}_{q^{2n}}[X] /\langle \mathrm{tr}_{2r}^s (X) \rangle$ (see Remark \ref{rem:compa}) and we will only need to evaluate monomials of degree less than $q^{2n}-1$.

\begin{de}
\label{TREC}
Let $\mathcal{H} = \{0,1, \ldots, q^{2n} -2\}$ and for any non-empty subset $\Delta \subseteq \mathcal{H}$, we define the evaluation code $E_{\Delta, \mathrm{tr}_{2r}^s}$ in $\mathbb{F}_{q^{2n}}^N$, as the linear code generated by the set of vectors $\{ \mathrm{ev}_{\mathrm{tr}_{2r}^s} (X^a) | a \in \Delta\}$.% which are in $\mathbb{F}_{q^{2n}}^N$.
\end{de}

\begin{pro}
\label{13}
Assume that $\Delta \subseteq \{0,1, \ldots, q^{2n-1} -1\}$. Then the dimension of the code $E_{\Delta, \mathrm{tr}_{2r}^s}$ coincides with the cardinality of the set $\Delta$.
\end{pro}
\begin{proof}
A generator matrix of the code consists of some rows of a Vandermonde matrix over the field $\mathbb{F}_{q^{2n}}$. These rows are linearly independent because $p^{2r-1}$ is the degree of the polynomial $\mathrm{tr}_{2r}^s(X)$ and $q^{2n-1}-1$ is the maximum degree of the involved monomials.
\end{proof}

Stabilizer quantum codes can be constructed from classical self-orthogonal codes with respect to the Hermitian inner product. Recall that the Hermitian inner product of two vectors $\mathbf{a} = (a_1, a_2, \ldots a_N)$ and $\mathbf{b} = (b_1, b_2, \ldots b_N)$ in $\mathbb{F}_{q^{2n}}^N$ is defined as
\[
\mathbf{a} \cdot_h \mathbf{b} := \sum_{j=1}^N a_j b_j^{q^n}.
\]
Hence, we will look for self-orthogonal codes $E_{\Delta, \mathrm{tr}_{2r}^s}$ with respect to this inner product, that is codes which satisfy
\[
 E_{\Delta, \mathrm{tr}_{2r}^s} \subseteq \left( E_{\Delta, \mathrm{tr}_{2r}^s}\right)^{\perp_h} := \left\{ \mathbf{b} \in \mathbb{F}_{q^{2n}}^N |  \mathbf{a} \cdot_h \mathbf{b}= 0 \mbox{ for all } \mathbf{a} \in E_{\Delta, \mathrm{tr}_{2r}^s} \right\}.
\]

The Euclidean inner product will be used in our development as well. For $\mathbf{a}$ and $\mathbf{b}$ in $\mathbb{F}_{q^{2n}}^N$, it is defined as $\mathbf{a} \cdot \mathbf{b} := \sum_{j=1}^N a_j b_j$. We start with a lemma which will allow us to derive the first result on the orthogonality of the generators of our codes.

\begin{lem}
\label{Newton}
Let $f$ be a polynomial in $\mathbb{F}_{q^{2n}} [X]$ of degree $m$, $f= \sum_{j=1}^m a_j X^j$ with $a_m=1$. Let $\{x_1, x_2, \ldots, x _m\}$ be the set of  roots of $f$ in $\mathbb{F}_{q^{2n}}$. Denote by $s_k$, $1 \leq k \leq m$, the power sum $s_k = \sum_{j=1}^m x_j^k$. Then
\begin{equation}
\label{estrella}
\left(  \sum_{j=0}^{i-1} a_{m-j} s_{i-j} \right) + i a_{m-i} =0,
\end{equation}
when $i \leq m$. Otherwise ($i>m$), it holds
\[
\sum_{j=0}^{m-1} a_{m-j} s_{i-j} =0.
\]
\end{lem}
\begin{proof}
It suffices to consider that the elementary symmetric elements $\sigma_k$, $1 \leq k \leq m$:
\[
\sigma_k = \sum_{i_1 < i_2 < \cdots < i_k} x_{i_1} x_{i_2} \cdots x_{i_k}
\]
and the Newton identities \cite[proof of Theorem 8 in Chapter 7, Section 1]{CLO} prove that
\[
s_k + \sum_{i=1}^{k-1} (-1)^i \sigma_i s_{k-i}  + (-1)^k k \sigma_k = 0,
\]
when $1 \leq k \leq m$. Moreover, for $k > m$,
\[
s_k + \sum_{i=1}^{m} (-1)^i \sigma_i s_{k-i}   = 0.
\]
Finally, the result holds since $a_j = (-1)^{m-j} \sigma_{m-j}$ \cite[Problem 4 in Chapter 7, Section 1]{CLO}.
\end{proof}
\begin{pro}
\label{16}
With  the above notations, recall that $p^{2r}  =q^{2n}$, one has that
\[
\mathrm{ev}_{\mathrm{tr}_{2r}^s} (X^k) \cdot \mathrm{ev}_{\mathrm{tr}_{2r}^s} (X^0) =0,
\]
for $1 \leq k < q^{2n-1} -1$ and
\[\mathrm{ev}_{\mathrm{tr}_{2r}^s} (X^{q^{2n-1} -1}) \cdot \mathrm{ev}_{\mathrm{tr}_{2r}^s} (X^0) \neq 0.
\]
\end{pro}
\begin{proof}
This result is a consequence of Lemma \ref{Newton}. Namely, notice that, with the notation as in Lemma \ref{Newton},  $\mathrm{ev}_{\mathrm{tr}_{2r}^s} (X^k) \cdot \mathrm{ev}_{\mathrm{tr}_{2r}^s} (X^0) = s_k$, where one shall consider the polynomial $\mathrm{ev}_{\mathrm{tr}_{2r}^s}$ instead of $f$ and $N$ instead of $m$. In addition, all the coefficients $a_j$ are equal to zero, but $a_1, a_q, a_{q^2}, \ldots, a_{q^{2n -1}}$ which are equal to $1$. Now Formula (\ref{estrella}) with $i=1$ proves that $s_1 = -a_{N-1} =0$; with $i=2$, $s_2= -2 a_{N-2} =0$, and iterating the same argument for consecutive values, one has that $s_k =0$ for indices $1 \leq k < q^{2n-1} - q^{2n-2}$. Again Formula (\ref{estrella}), for $i= q^{2n-1} - q^{2n-2}$, proves that $s_{q^{2n-1}- q^{2n-2}}=0$ since we work over a field of characteristic $p$. It is clear that the same procedure proves that $s_k =0$ for $1 \leq k < q^{2n-1} - 1$.

Finally $s_{q^{2n-1} - 1} \neq 0$, because Formula (\ref{estrella}) for $i=q^{2n-1} - 1$ shows that
\[
s_{q^{2n-1} - 1} + a_{q^{2n-1} - 1} s_{q^{2n-1} - 2}+ \cdots + a_1 (q^{2n-1} - 1) =0,
\]
and then $s_{q^{2n-1} - 1} = - ( q^{2n-1} - 1) = 1 \neq 0$, which concludes the proof.
\end{proof}

The map $\mathrm{ev}_{\mathrm{tr}_{2r}^s}$ is defined for elements in $ \mathbb{F}_{q^{2n}}[X] /\langle \mathrm{tr}_{2r}^s (X) \rangle$ which have as class representative polynomials of degree lower than $q^{2n-1}$. Proposition \ref{16} shows that the evaluation by $\mathrm{ev}_{\mathrm{tr}_{2r}^s}$ of a (class of a) polynomial $f$ in $\mathbb{F}_{q^{2n}}[X]$ is Euclidean orthogonal to $\mathrm{ev}_{\mathrm{tr}_{2r}^s} (X^0)$ if and only if the mentioned representative does not contain the monomial $X^{q^{2n-1}-1}$. This proves the following result which complements Proposition \ref{16}.

\begin{pro}
\label{17}
With the above notation, for $k \le q^{2n}-2$, the Euclidean inner product
\[
\mathrm{ev}_{\mathrm{tr}_{2r}^s} (X^k) \cdot \mathrm{ev}_{\mathrm{tr}_{2r}^s} (X^0) = 0
\]
if and only if the polynomial of degree less than $q^{2n-1}$ representing the class $X^k + \langle \mathrm{tr}_{2r}^s (X) \rangle$ does not contain the monomial $X^{q^{2n-1}-1}$.
\end{pro}

Next, we give a condition implying that some classes as above do not contain $X^{q^{2n-1}-1}$ in their representatives.

\begin{pro}
\label{18}
With the above notation, let $i,j$ be integers such that
\[
0 \leq i, j<q^n-\lfloor\frac{(q-1)}{2}\rfloor q^{n-1}-\cdots-\lfloor\frac{(q-1)}{2}\rfloor q-1,
\]
which are not both zero. Then, for $0 < m \leq n$, the representative of the class $X^{i+j q^m} +  \langle \mathrm{tr}_{2r}^s (X) \rangle$ of degree less than $q^{2n-1}$ does not contain the monomial $X^{q^{2n-1}-1}$.
\end{pro}

\begin{proof}
Write $\delta = q - \lfloor\frac{(q-1)}{2}\rfloor$ and notice that $\delta = \frac{(q+1)}{2}$ if $q$ is odd and it equals $\frac{(q+2)}{2}$ otherwise. Thus, the bound $q^n-\lfloor\frac{(q-1)}{2}\rfloor q^{n-1}-\cdots-\lfloor\frac{(q-1)}{2}\rfloor q-1$ can be expressed as
\begin{equation}
\label{dosestrellas}
\delta q^{n-1} -\lfloor\frac{(q-1)}{2}\rfloor q^{n-2}-\cdots-\lfloor\frac{(q-1)}{2}\rfloor q-1.
\end{equation}
Now, consider the $q$-adic expansion of $i$ and $j$:
\[
i= \sum_{k=0}^{n-1} a_k q^k, \;\;\; j= \sum_{k=0}^{n-1} b_k q^k.
\]
For  $i$ (and analogously for $j$), the expression in (\ref{dosestrellas}) shows that:

\begin{itemize}
\item {\it  When $q$ is even}, $a_{n-1} \leq \delta -1$ and when $a_{n-1} = \delta -1$, then $a_{n-2} \leq \delta -1$, fact that we can iterate and claim that $a_0 \leq \delta -1$, whenever $a_1 = a_2 = \cdots = a_{n-1} = \delta -1$. There exists an exception for $q=2$, in this case $\delta=2$ and $a_0 =0$, whenever $a_1 = a_2 = \cdots = a_{n-1} = 1$.

\item {\it  Otherwise ($q$ is odd)}, one also has that $a_{n-1} \leq \delta -1$. If $a_{n-1} = \delta -1$, then $a_{n-2} \leq \delta -1$ and, as above, this argument can be repeated and one gets that $a_0 \leq \delta$, when $a_1 = a_2 = \cdots = a_{n-1} = \delta - 1$.
\end{itemize}

We divide our reasoning in two cases:

{\it Case 1, $m < n$:} then $n -1 = m + m_1$, where $m_1 \geq 0$. Then
\[
i+j q^m= a_0+a_1q+\cdots+(a_m+b_0)q^m+\cdots+(a_{n-1}+b_{m_1})q^{n-1}
\]
\[
+b_{m_1+1}q^{n}+\cdots+b_{n-1}q^{n+m-1} \leq 2 q^{n} +b_{m_1 +1}q^{n} + \cdots + b_{n-1} q^{n+m-1} \leq
\]
\[
(b_{n-1} +1)  q^{n+m-1} < q^{2n -1} -1,
\]
the last inequality holds because otherwise $m=n-1$ (notice that $m<n$) and $b_{n-1} + 1 = q$ and then
\[
i +  j q^m = a_0 + \cdots + (a_{n-1} + b_0) q^{n-1} + b_1 q^n + \cdots + b_{n-1} q^{2n-2}.
\]
The last expression is equal to $q^{n-1} -1$ only when all the coefficients are exactly equal to $q-1$, which gives a contradiction because $a_0 \leq \delta$ as we indicated previously.

{\it Case 2, $m=n$:} then,
\[
i + j q^m = i + j q^n = a_0 + a_1 q + \cdots + a_{n-1} q^{n-1} + b_0 q^n + b_1 q^{n+1} +  \cdots + b_{n-1} q^{2n-1}.
\]
This expression is the exponent of  a term in $X$ which can be written as
\begin{equation}\label{cruzcuadro}
X^{a_0 + a_1 q + \cdots + b_{n-2} q^{2n-2} } ( X^{q^{2n-1}})^{b_{n-1}}.
\end{equation}
Since we are considering the class of the term in $\mathbb{F}_{q^{2n}}[X] /\langle \mathrm{tr}_{2r}^s (X) \rangle$, we can replace the monomial $X^{q^{2n-1}}$ with the polynomial $-X-X^q - \cdots - X^{q^{2n-2}}$. The multinomial theorem shows that the expression in (\ref{cruzcuadro}) can be expressed as a sum of terms where the exponents of the attached monomials are of the form
\[
a_0 + a_1 q + \cdots + a_{n-1} q^{n-1} + b_0 q^n+ \cdots + b_{n-2} q^{2n-2} + \sum_{k=0}^{2n-2} c_k q^k.
\]
Notice that $\sum_{k=0}^{2n-2} c_k q^k$ is the $q$-adic expansion of the exponent of some monomial in
\begin{equation}
\label{T}
(-X-X^q - \cdots - X^{q^{2n-2}})^{b_{n-1}}
\end{equation}
and therefore $\sum_{k=0}^{2n-2} c_k = b_{n-1} \leq \delta -1$. As a consequence, we get terms whose exponents (of the corresponding monomials) are
\begin{equation}
\label{Z}
\sum_{k=0}^{n-1} (a_k + c_k ) q^k + \sum_{k=0}^{n-2} (b_k + c_{k+n}) q^{k+n}.
\end{equation}

Consider first the case when $q$ is odd. Then, for having a term whose monomial is $X^{q^{2n-1}-1}$,  every coefficient in the $q$-adic expansion of (\ref{Z}) shall be equal to $q-1$. As $b_k$ and $c_k$ are lower than $\delta = (q+1)/2$, it holds that $b_k + c_{k+n} \leq q-1$. However, $b_{n-2} + c_{2n-2}$ is the coefficient of $q^{2n-2}$ and it equals $q-1$ only when $b_{n-1} = (q-1)/2$ and uniquely for one monomial obtained from (\ref{T}), but in this case $c_{2n-3}=0$, and thus not all coefficients in (\ref{Z}) are equal to $q-1$.

Finally, when $q$ is even, $\delta =(q+2)/2 = q/2 + 1$ and then the sums $a_k + c_k$, $0 \leq k \leq n-1$ and $b_k + c_{k+n}$, $0 \leq k \leq n-2$, may reach the values $q-1$ or $q$. However, this is not the case for $a_0 + c_0$ because $c_0$ is either $0$ or $1$ depending on either $b_{n-1} >1$ or $b_{n-1} =1$. When either $a_k + c_k$, for $0 \leq k \leq n-1$, or $b_k + c_{k+n}$, for $0 \leq k \leq n-2$, is equal to $q$, the $q$-adic expansion of (\ref{Z}) is obtained by adding one unit to the next power of $q$, and when $b_{n-2} + c_{2n-2} =q$, again one must use the fact that $X^{q^{2n-1}} = -X-X^q - \cdots - X^{q^{2n-2}}$. Taking into account that the power $(X^{q^{2n-1}})^i$ with $i=1$ can appear only once, we deduce that the $q$-adic expansion $\sum_{k=0}^{2n-2} d_k q^k$ of the expression (\ref{Z}) satisfies $d_k < (\delta -1)+1 = (q+2)/2 < q-1$ and not every coefficient of the mentioned $q$-adic expansion is equal to $q-1$.
\end{proof}

We conclude this section with a result which gives the parameters of the quantum codes constructed from Hermitian duals of certain codes $E_{\Delta, \mathrm{tr}_{2r}^s}$. These codes are MDS quantum codes and they were also found in \cite{MDS1,Sarvepalli}.

\begin{teo}
\label{19}
Let $p$ be a prime number, $r$ and $s$ positive integers such that $r=s \cdot n$, $n \geq 1$ and set $q= p^{s}$. Let $t$ be a nonnegative integer such that
\[
t <q^n-\lfloor\frac{(q-1)}{2}\rfloor q^{n-1}-\cdots-\lfloor\frac{(q-1)}{2}\rfloor q-1
\]
and write $\Delta(t) = \{ a \in \mathbb{Z} ~|~ 0 \leq a \leq t \}$.
Then, the following inclusion holds:
\[
E_{\Delta(t), \mathrm{tr}_{2r}^s} \subseteq \left( E_{\Delta(t), \mathrm{tr}_{2r}^s}\right)^{\perp_h}.
\]

As a consequence, we are able to construct a stabilizer (quantum) MDS code with parameters $[[N,N -2t -2,  t+2]]_{q^n}$.
\end{teo}
\begin{proof}
Propositions  \ref{17} and \ref{18} for $m=n$ show that
\[
\mathrm{ev}_{\mathrm{tr}_{2r}^s} (X^i) \cdot_h \mathrm{ev}_{\mathrm{tr}_{2r}^s} (X^j) = \mathrm{ev}_{\mathrm{tr}_{2r}^s} (X^{i+ j q^n}) \cdot \mathrm{ev}_{\mathrm{tr}_{2r}^s} (X^0) =0,
\]
where the monomials $X^i$ and $X^j$ are representatives of  classes in $\mathbb{F}_{q^{2n}}[X] /\langle \mathrm{tr}_{2r}^s (X) \rangle$ and $i, j \in \Delta(t)$.
This proves the codes' inclusion. The dimension of the  stabilizer code is clear from Proposition \ref{13} and Theorem \ref{el1}. Finally, we use Theorem \ref{el1} again for bounding the distance of the stabilizer code. Indeed, by Proposition \ref{16} the code $\left( E_{\Delta(t), \mathrm{tr}_{2r}^s}\right)^{\perp}$ contains the image by $\mathrm{ev}_{\mathrm{tr}_{2r}^s}$ of consecutive monomials $X^j$, $0 \leq   (N-1)-(t+1)$,  because $E_{\Delta(t), \mathrm{tr}_{2r}^s}$ is the code generated by $\mathrm{ev}_{\mathrm{tr}_{2r}^s}(X^i)$, $0 \leq i \leq t$. Thus, the minimum distance of the code is at least $t+2$ but it cannot be larger than the Singleton bound. This concludes the proof after noticing that Hermitian and Euclidean dual codes are isometric, which can be deduced from the fact that, in our case, the Euclidean dual of a code coincides with  the $q^n$th power of its Hermitian  dual.
\end{proof}

\section{Subfield-subcodes of evaluation codes}\label{sect:subfield}

In this section, we will review and adapt to our notation known results on subfield subcodes of evaluation codes. We will follow the approach in \cite{galindo-hernando, gal-her-rua, QINP, QINP2} to obtain subfield-subcodes, namely, we will consider subfield subcodes of one-variable $J$-affine variety codes with $J=\emptyset$. We refer the reader to these references for proofs and further details.

We recall that $p$ is a prime number and $r$ and $s$ are positive integers such that $s|r$. Let $N^T= p^{2r}$ and consider the map $\mathrm{ev}^T: \mathbb{F}_{p^{2r}}[X]/\langle X^{N^T} - X \rangle \rightarrow \mathbb{F}_{p^{2r}}^{N^T}$ defined by $\mathrm{ev}^T(f) = (f(\alpha_1), f(\alpha_2), \ldots, f(\alpha_{N^T}))$, where $Z^T = \{\alpha_1, \alpha_2, \ldots, \alpha_{N^T}\}$ is the set of zeros of the polynomial $X^{N^T} - X $ in $\mathbb{F}_{p^{2r}}$. Note that $Z \subset  Z^T$ by Section \ref{sect1}. Let $\Delta \subseteq \{0, 1 , \ldots , N^T -1 \}$, we define the evaluation code $E^T_\Delta \subseteq \mathbb{F}_{p^{2r}}$ as the linear space generated by the vectors $\{ \mathrm{ev}^T (X^a) \mid a \in \Delta \}$. For  $\Delta = \{ 0, 1, \ldots, k-1 \}$ we have a Reed-Solomon code with length $p^{2r}$ and dimension $k$. In general, the dimension of $E^T_\Delta$ is equal to the cardinality of the set $\Delta$.

Let $\mathcal{H}^T = \{0\} \cup \{ 1, 2, \ldots, N^T -1 \}$, where $\{ 1, 2, \ldots, N^T -1 \}$ is regarded as a set of representatives of the congruence  ring $\mathbb{Z}_{N^T -1} = \mathbb{Z}/(N^T -1) \mathbb{Z}$, and consider cyclotomic cosets with respect to $p^{2s}$ defined as subsets $\mathfrak{I} \subseteq \mathcal{H}^T$ such that $p^{2s} a \in \mathfrak{I}$ for all $a \in \mathfrak{I}$. A cyclotomic coset $\mathfrak{I}$ as above is said to be {\it minimal} whenever its elements are those that can be expressed as $a p^{(2s)i}$, for some nonnegative integer $i$ and some fixed element $a \in \mathfrak{I}$. We represent each minimal cyclotomic coset $\mathfrak{I}$ by that element $a$ in $\mathcal{H}^{T}$ which is the minimum in $\mathfrak{I}$ and then we write $\mathfrak{I}= \mathfrak{I}_a$. This set of representatives will be denoted by $\mathcal{A}$ and so $\{\mathfrak{I}_a\}_{a \in \mathcal{A}}$ is the family of minimal cyclotomic cosets in $\mathcal{H}^T$.

Next, we consider a different trace map, $\mathrm{tr}_{2r}^{2s}: \mathbb{F}_{p^{2r}} \rightarrow \mathbb{F}_{p^{2s}}$, defined as
\[
\mathrm{tr}_{2r}^{2s}(x) = x + x^{p^{2s}}+ \cdots + x^{p^{2s(\frac{r}{s}-1)}},
\]
%which componentwise determines $\mathbf{tr}_{2r}^{2s}: (\mathbb{F}_{p^{2r}})^{N^T} \rightarrow (\mathbb{F}_{p^{2s}})^{N^T}$.
and let
\[
\mathcal{T}: \mathbb{F}_{p^{2r}}[X]/\langle X^{N^T} - X \rangle \rightarrow \mathbb{F}_{p^{2r}}[X]/\langle X^{N^T} - X \rangle,
\]
given by $\mathcal{T}(f) = f + f^{p^{2s}}+ \cdots + f^{p^{2s(\frac{r}{s}-1)}}$. This last map satisfies the following result whose proof is identical to that of \cite[Proposition 5]{galindo-hernando}.
\begin{pro}
\label{5DCC}
Let $f$ be an element in $\mathbb{F}_{p^{2r}}[X]/\langle X^{N^T} - X \rangle$. Then, the following conditions are equivalent:
\begin{enumerate}
\item $f = \mathcal{T}(h)$ for some $h \in \mathbb{F}_{p^{2r}}[X]/\langle X^{N^T} - X \rangle$.
    \item $f^{p^{2s}} = f$.
    \item $f$ evaluates to $\mathbb{F}_{p^{2s}}$, that is $\mathrm{ev}^T (f) \in (\mathbb{F}_{p^{2s}})^{N^T}$.
\end{enumerate}
\end{pro}
The above result shows that one can get codes of length $N^T$ over $\mathbb{F}_{p^{2s}}$ from the images $\mathrm{ev}^T(\mathcal{T}(h))$ of  classes of polynomials $h \in \mathbb{F}_{p^{2r}}[X]$.

Next we provide a result very close to \cite[Theorem 3]{galindo-hernando}, whose proof is analogous, which determines a basis of the vector space (over $\mathbb{F}_{p^{2s}}$) of polynomials in $\mathbb{F}_{p^{2r}}[X]/\langle X^{N^T} - X \rangle$ evaluating to $\mathbb{F}_{p^{2s}}$. In order to state such a result,  we need the following notation: $i_a$ denotes the cardinality of the minimal cyclotomic coset $\mathfrak{I}_a$ and, since $2s i_a$ divides $2r$,  the mapping for polynomials $f$ with support on a cyclotomic coset $\mathfrak{I}_a$
\[
\mathcal{T}_a(f) = f + f^{p^{2s}}+ \cdots + f^{p^{2s(i_a-1)}},
\]
is well defined.

\begin{pro}
\label{23}
With the above notation, it holds that the set
\[
\bigcup_{a \in \mathcal{A}} \left\{ \mathcal{T}_a \left( \beta^l X^a \right) \; \bigl\vert \; 0 \leq l \leq i_a -1 \mbox{ and $\beta$ is a primitive element of $\mathbb{F}_{p^{2s i_a}}$ } \right\}
\]
is a basis of the vector space (over $\mathbb{F}_{p^{2s}}$) of elements in $\mathbb{F}_{p^{2r}}[X]/\langle X^{N^T} - X \rangle$ evaluating to $\mathbb{F}_{p^{2s}}$
\end{pro}

Let $E^{T,\sigma}_\Delta$ be the subfield subcode of $E^T_\Delta$ over $\mathbb{F}_{p^{2s}}$, i.e.  $E^{T,\sigma}_\Delta = E^{T,\sigma}_\Delta \cap \mathbb{F}_{p^{2s}}$. By \cite[Theorem 4]{galindo-hernando}, the dimension of $E^{T,\sigma}_\Delta$ is equal to $$ \sum_{\mathbf{a} \in \mathcal{A}| \mathfrak{I}_\mathbf{a} \subseteq \Delta} i_\mathbf{a}.$$

Let $C^{T,\sigma}_\Delta$ be the Euclidean dual code of $E^{T,\sigma}_\Delta$,  and $\mathcal{A}= \{ a_0=0 <a_1<a_2 \cdots <a_z\}$, for $t \leq z$. For $\Delta^\sigma(t) = \mathfrak{I}_{a_0} \cup \mathfrak{I}_{a_1} \cup \cdots \cup \mathfrak{I}_{a_t}$, the minimum distance of $C^{T,\sigma}_{\Delta^\sigma(t)}$ is greater than or equal to $a_{t+1}+1$ (BCH bound).

\begin{exa}\label{ex:t1}
Let $p=2$, $s=1$ and $r=4$. Hence, we will consider codes over $\mathbb{F}_{2^8}$ and subfield-subcodes over $\mathbb{F}_{2^2}$ with length $N^T=256$. The first eight minimal cyclotomic cosets are $I_0 = \{0\}$, $I_1 = \{1,4,16,64\}$, $I_2=\{2,8,32,128\}$, $I_3=\{3,12,48,142\}$, $I_5=\{5,20,65,80\}$, $I_6 = \{6,24,12, 129\}$, $I_7=\{7,28,112,193\}$ and $I_9= \{9,36,66,144\}$. Hence we have that $a_0 = 0$, $a_1=1$, $a_2=2$, $a_3=3$, $a_4 = 5$, $a_5 = 6$, $a_6 =7$, $a_7=9$.

Consider $\Delta^\sigma(6) = \mathfrak{I}_{a_0} \cup \mathfrak{I}_{a_1} \cup \cdots \cup \mathfrak{I}_{a_6}$ . The code  $C^{T,\sigma}_{\Delta^\sigma(6)}$ has parameters $$\left[N^T,N^T- \sum_{l=0}^6 i_{a_l},a_7+1\right]_4=[256, 256-25, \ge 10]_4 = [256, 231, \ge 10]_4.$$
\end{exa}

\section{Stabilizer codes obtained from subfield-subcodes of Evaluation Codes at the Trace Roots}
\label{sect2}

The aim of this section is to study subfield-subcodes over $\mathbb{F}_{p^{2s}}$ of the codes introduced in Section \ref{sect1} and determine the parameters for their attached stabilizer quantum codes over $\mathbb{F}_{p^s}$. Keep the notation as in that section.

%Our codes are defined by the map $\mathrm{ev}_{\mathrm{tr}_{2r}^s}$; to study them we recall some known facts that hold when one evaluates in more points than we have used above.

%As mentioned, we would like to compute subfield-subcodes of our codes for obtaining stabilizer quantum codes. We start by defining the classical codes that support this fact.

\begin{de}
Let $\emptyset \neq \Delta \subseteq \mathcal{H}$, the subfield-subcode over $\mathbb{F}_{p^{2s}}$ of the code $E_{\Delta, \mathrm{tr}_{2r}^s}$ is defined as
\[
E_{\Delta, \mathrm{tr}_{2r}^s}^\sigma := E_{\Delta, \mathrm{tr}_{2r}^s} \cap \mathbb{F}_{p^{2s}}^N.
\]
\end{de}

%\begin{rem}
%{\rm
%}
%\end{rem}

Proposition \ref{5DCC} and the paragraph before Proposition \ref{23} prove that the map $\mathrm{ev}_{\mathrm{tr}_{2r}^s}$ applied to classes of polynomials $\mathcal{T} (f)$ (and $\mathcal{T}_a (f)$ )  that evaluate  to  $\mathbb{F}_{p^{2s}}^N$, where $N=q^{2n-1} = p^{2r-s}$. Moreover, considering suitable sets $\Delta$, we can bound their parameters. Let $\mathcal{A}= \{ a_0=0 <a_1<a_2 \cdots <a_z\}$  and, for $t \leq z$, let \[
\Delta^\sigma (t) := \mathfrak{I}_{a_0} \cup \mathfrak{I}_{a_1} \cup \cdots \cup \mathfrak{I}_{a_t}.
\]
Then,
\begin{teo}
\label{25}
%Let $E_{\Delta^\sigma (t)}$ be the code generated by $\{\mathrm{ev}^T (X^a) | a \in \Delta^\sigma (t)\}$ and $E_{\Delta^\sigma (t)}^\sigma = E_{\Delta^\sigma (t)} \cap \mathbb{F}_{p^{s}}^N$ a subfield-subcode.

The dimension of $E_{\Delta^\sigma (t), \mathrm{tr}_{2r}^s}^\sigma$ and the minimum distance of its Hermitian dual code satisfy the following bounds:
\[
\dim \left( E_{\Delta^\sigma (t), \mathrm{tr}_{2r}^s}^\sigma \right) \leq \sum_{l=0}^t i_{a_l},
\]
\[
 d \left( E_{\Delta^\sigma (t), \mathrm{tr}_{2r}^s}^\sigma \right)^{\perp_h} \geq a_{t+1} +1
\]
\end{teo}
\begin{proof}

By \cite[Theorem 4]{galindo-hernando}, we have that $\dim \left( E_{\Delta^\sigma (t)}^{T,\sigma} \right) = \sum_{l=0}^t i_{a_l}$. Here, since we only evaluate at the zeros of the $\mathrm{tr}_{2r}^s(X)$ ($Z \subset Z^T$), the first inequality holds.

With respect to the last inequality, setting $A = \{0,1, \ldots, a_{t+1} -1\}$, it holds that $A \subseteq \Delta^\sigma (t)$ and then one gets the inclusion of codes in $\mathbb{F}_{p^{2r}}$: $E_{A,\mathrm{tr}_{2r}^s} \subseteq E_{\Delta^\sigma (t), \mathrm{tr}_{2r}^s}$. Thus, the Euclidean dual of both codes satisfy $(E_{\Delta^\sigma (t),\mathrm{tr}_{2r}^s})^\perp \subseteq (E_{A,\mathrm{tr}_{2r}^s})^\perp$. Therefore,
\[
d \left((E_{\Delta^\sigma (t), \mathrm{tr}_{2r}^s})^\perp \right) \geq d \left( E_{A, \mathrm{tr}_{2r}^s}^\perp \right) \geq a_{t+1} +1,
\]
because the parity check matrix of $E_{A, \mathrm{tr}_{2r}^s}^\perp $ corresponds with the generator matrix of $E_{A, \mathrm{tr}_{2r}^s}$, which is a Vandermonde matrix. Considering subfield-subcodes over $\mathbb{F}_{p^{2s}}$, we have that
\[
\left( E_{\Delta^\sigma (t), \mathrm{tr}_{2r}^s}^\sigma \right)^{\perp} = \left( E_{\Delta^\sigma (t), \mathrm{tr}_{2r}^s}^\perp \right)^{\sigma} \subseteq \left(E_{A, \mathrm{tr}_{2r}^s}^\perp \right)^\sigma,
\]
where the equality follows from Delsarte Theorem \cite{delsarte}. Then,
\[
d\left( E_{\Delta^\sigma (t), \mathrm{tr}_{2r}^s}^\sigma \right)^{\perp} = d\left( E_{\Delta^\sigma (t), \mathrm{tr}_{2r}^s}^\perp \right)^{\sigma} \geq d\left(E_{A, \mathrm{tr}_{2r}^s}^\perp \right)^\sigma \geq a_{t+1} +1.
\]
This concludes the proof because the Euclidean and Hermitian dual of our codes are isometric.
\end{proof}

\begin{exa}\label{ex:t2}
Let $p=2$, $s=1$ and $r=4$. We will consider a code over $\mathbb{F}_{2^8}$ and a subfield-subcode over $\mathbb{F}_{2^2}$ as in Example \ref{ex:t1}. We have that $N=128$ and consider again $\Delta^\sigma(6) = \mathfrak{I}_{a_0} \cup \mathfrak{I}_{a_1} \cup \cdots \cup \mathfrak{I}_{a_6}$.
The code $\left( E_{\Delta^\sigma (6), \mathrm{tr}_{2r}^s}^\sigma \right)^{\perp_h}$ has parameters $$\left[N, \ge N - \sum_{l=0}^6 i_{a_l} ,a_7+1\right]_4 = [128,\ge 128-25,\ge 10]_4= [128,\ge 103,\ge 10]_4.$$

Moreover, we know that the dimension is strictly greater than $103$ since $\mathcal{T}_1(X)$ and $\mathcal{T}_2 (X)$ are equal modulo $\mathrm{tr}^1_{8} (X)$, because $\mathcal{T}_1(X) = X + X^4 + X^{16} + X^{64}$, $\mathcal{T}_2 (X) = X^2 + X^8 + X^{32} + X^{128}$, and
$\mathrm{tr}^1_{8} (X) = X + X^2 +  X^4 + X^8 + X^{16} + X^{32} + X^{64} +X^{128}$. Actually one can prove that the code $\left( E_{\Delta^\sigma (6), \mathrm{tr}_{2r}^s}^\sigma \right)^{\perp_h}$ has parameters $[128,104,10]_4$.
\end{exa}

\begin{rem}\label{rem:compa}
Examples \ref{ex:t1} and \ref{ex:t2} help to illustrate how to compare the codes obtained in the previous section --extended BCH codes (or subfield-subcodes of $J$-affine codes with $J = \emptyset$)--  with subfield-subcodes of evaluation codes at the trace roots. When considering dual codes, the advantage of the last code can be observed from the difference between the length and dimension since both codes have the same designed minimum distance. First observe that such a difference is equal to $\sum_{l=0}^t i_{a_l}$ in both cases (25 in our examples), however for the evaluation codes at the trace roots we have an advantage: their dimension may be strictly greater than the designed dimension  $N - \sum_{l=0}^t i_{a_l}$, as the previous example shows. This will allow us to get classical and quantum codes with excellent parameters. In general, there may be several relations modulo $\mathrm{tr}^s_{2r} (X)$ among the polynomials in Proposition \ref{23}, which increase the dimension of $\left( E_{\Delta^\sigma (t), \mathrm{tr}_{2r}^s}^\sigma \right)^{\perp_h}$.
\end{rem}

We conclude this section with our main result that shows how to construct stabilizer codes from subfield-subcodes over $\mathbb{F}_{p^{2s}}$. Recall that  $q=p^s$.

\begin{teo}
\label{26}
Let $N=q^{2n-1}$ the degree of the polynomial $\mathrm{tr}_{2r}^s (X)$, $N^T=p^{2r}$ and $\mathcal{A}= \{ a_0=0 <a_1<a_2 \cdots <a_z\}$ the set of representatives of the minimal cyclotomic sets $\mathfrak{I}_{a_i}$, $0 \leq i \leq z$ of $\mathcal{H}^T$ with respect to $p^{2s}$. Let $t \le z$ be an index such that
\[
a_{t}<q^n-\lfloor\frac{(q-1)}{2}\rfloor q^{n-1}-\cdots-\lfloor\frac{(q-1)}{2}\rfloor q-1.
\]
Then, with the notation as above, the following  inclusion holds
\begin{equation}
\label{261}
E_{\Delta^\sigma (t), \mathrm{tr}_{2r}^s}^\sigma \subseteq \left( E_{\Delta^\sigma (t), \mathrm{tr}_{2r}^s}^\sigma \right)^{\perp_h},
\end{equation}
where $\Delta^\sigma (t) = \mathfrak{I}_{a_0} \cup \mathfrak{I}_{a_1} \cup \cdots \cup \mathfrak{I}_{a_t}$.

As a consequence, we are able to construct a stabilizer code with parameters
\[
\left[\left[N, \geq N - 2 \sum_{a=0}^t i_a , \geq a_{t+1} +1 \right]\right]_q.
\]
\end{teo}
\begin{proof}
By Theorem \ref{25}, it suffices to prove the inclusion in (\ref{261}). We shall show that
\begin{equation}
\label{262}
\mathrm{ev}_{\mathrm{tr}_{2r}^s} \left( \mathcal{T}_{a_i} (\beta_1^{k_1} X^{a_i}) \right) \cdot_h \mathrm{ev}_{\mathrm{tr}_{2r}^s} \left( \mathcal{T}_{a_j} (\beta_2^{k_2} X^{a_j}) \right) = 0,
\end{equation}
for $\beta_1$ (respectively, $\beta_2$) a primitive element in $\mathbb{F}_{p^{2 s i_{a_i}}}$ (respectively, in $\mathbb{F}_{p^{2 s i_{a_j}}}$), for $0 \leq k_1 \leq i_{a_i} -1$ (respectively, for $0 \leq k_2 \leq i_{a_j} -1$) and $i, j \in \{0, 1, \ldots, t\}$. This will conclude the proof by Proposition \ref{23}.

The left hand side in (\ref{262}) is a summation, up to constants that depend on $\beta_1$ and $\beta_2$, of Euclidean products of the form
\begin{equation}
\label{263}
\mathrm{ev}_{\mathrm{tr}_{2r}^s} \left( X^{a  q^l + b  q  q^m} \right) \cdot \mathrm{ev}_{\mathrm{tr}_{2r}^s} \left( X^{0} \right),
\end{equation}
where $a, b$ are the corresponding representatives in $\mathcal{A}$. We can assume that
$a, b <q^n-\lfloor\frac{(q-1)}{2}\rfloor q^{n-1}-\cdots-\lfloor\frac{(q-1)}{2}\rfloor q-1$; and $l, m  \in \{0, 1, \ldots, 2n-1\}$.

We claim that each product of the form given in (\ref{263}) equals zero, which proves Equality (\ref{262}). Indeed, without loss of generality, we may assume that $m \geq l$ and divide the proof in two parts.

First, suppose that $m-l \leq n-1$. Then
\begin{equation}
\label{264}
\mathrm{ev}_{\mathrm{tr}_{2r}^s} \left( X^{a  q^l + b  q  q^m} \right) \cdot \mathrm{ev}_{\mathrm{tr}_{2r}^s} \left( X^{0} \right) = \left(
\mathrm{ev}_{\mathrm{tr}_{2r}^s} \left( X^{a  + b  q^{m-l+1}} \right) \cdot \mathrm{ev}_{\mathrm{tr}_{2r}^s} \left( X^{0} \right) \right)^{q^l},
\end{equation}
because of the characteristic of the field. Now, Proposition \ref{18} proves that the right hand side of Equality (\ref{264}) is equal to zero since $m -l +1 \leq n$, which concludes the first part.

Finally,  assume that $m-l \geq n$, then $l \leq m -n \leq  (2n -1)-n = n-1$ and $m= n + n_1 \leq 2n -1$, thus $n_1 < n$. In addition, Formula (\ref{263}) is equal to zero if and only if
\[
\left(\mathrm{ev}_{\mathrm{tr}_{2r}^s} \left( X^{a  q^l + b  q q^{n+n_1}} \right) \cdot \mathrm{ev}_{\mathrm{tr}_{2r}^s} \left( X^{0} \right) \right)^{q^n}
\]
is equal to zero. This last expression can also be written as
\[
\mathrm{ev}_{\mathrm{tr}_{2r}^s} \left( X^{a  q^{l+n} + b    q^{2n+n_1+1}} \right) \cdot \mathrm{ev}_{\mathrm{tr}_{2r}^s} \left( X^{0} \right).
\]
Since we are evaluating elements in the field $\mathbb{F}_{p^{2r}} = \mathbb{F}_{q^{2n}}$, it suffices to prove
\begin{equation}
\label{265}
\mathrm{ev}_{\mathrm{tr}_{2r}^s} \left( X^{a  q^{l+n} + b  q^{n_1+2}} \right) \cdot \mathrm{ev}_{\mathrm{tr}_{2r}^s} \left( X^{0} \right) =0,
\end{equation}
which holds whenever
\[
\left(\mathrm{ev}_{\mathrm{tr}_{2r}^s} \left( X^{a q^{l+n-n_1-2}+b} \right) \cdot \mathrm{ev}_{\mathrm{tr}_{2r}^s} \left( X^{0} \right)\right)^{q^{n_1+2}}
\]
is equal to zero. Note that this holds by Proposition \ref{18} since $l+n-n_1 -2 <n$. In fact, $n+n_1 -l >n > n-1$ and then $l-n_1 -1 <0$. This concludes the proof.
\end{proof}

\begin{exa}\label{ex:t3}
Let $p=2$, $s=1$, $r=4$, $n=4$ and $q=2$. Consider the classical subfield-subcode over $\mathbb{F}_4$, $ E_{\Delta^\sigma (6), \mathrm{tr}_{2r}^s}^\sigma$, given in Example \ref{ex:t2}. Since $a_6 = 7 < 15 = 2^4 -1 = q^n-\lfloor\frac{(q-1)}{2}\rfloor q^{n-1}-\cdots-\lfloor\frac{(q-1)}{2}\rfloor q-1$, we can apply Theorem \ref{26} and therefore it is self-orthogonal with respect to the Hermitian inner product. Its Hermitian dual has parameters $[128,104,10]_4$, therefore, by Theorem \ref{el1}, we obtain a stabilizer code  with parameters $[[128,2 \cdot 104 - 128 ,10]]_2 = [[128, 80,10]]_2$. This code is a record at \cite{codet} as we will see in Example \ref{ex:uno} in Section \ref{sect3}.
\end{exa}

To end this section, we consider another construction of linear codes: we have shown that $\mathrm{ev}_{\mathrm{tr}_{2r}^s}$ evaluates at the points in $Z$, which is a subset of the zero-set $Z^T$ of $X^{p^{2r}} -X$. By \cite[Proposition 1]{QINP}, Proposition \ref{17} also holds for $\mathrm{ev}^T$ when, as above,
\[
k <q^n-\lfloor\frac{(q-1)}{2}\rfloor q^{n-1}-\cdots-\lfloor\frac{(q-1)}{2}\rfloor q-1.
\]
Since $Z \subset Z^T$, for $Z^T \setminus Z = \{\gamma_1, \gamma_2, \ldots, \gamma_{N^C}\}$,  where $N^C = N^T -N$, and considering the evaluation map
\[
\mathrm{ev}^C: \frac{\mathbb{F}_{p^{2r}}[X]}{\left\langle (X^{N^T} - X) / \mathrm{ev}_{2r}^s (X) \right\rangle} \rightarrow \mathbb{F}_{p^{2r}}^{N^C},
\]
given by $\mathrm{ev}^C (f) = f(\gamma_1, \gamma_2, \ldots, \gamma_{N^c})$, one gets that, with the same reasoning, our results hold for  these linear and stabilizer quantum codes as well. We will refer to  these linear codes (respectively, their subfield-subcodes and the corresponding stabilizer codes) as {\it complementary codes} (respectively, their subfield-subcodes  and the stabilizer codes obtained from them).

\section{Examples}
\label{sect3}
In this section we give the parameters of a number of stabilizer codes obtained or derived from our development. First, we recall that Theorem \ref{26} shows how to use subfield-subcodes for constructing stabilizer codes over $\mathbb{F}_{q}$ with length $N=q^{2n-1}$, for $q=p^s$, where $p$ is a prime number and $s$ and $n$ are positive integers. The same reasoning gives rise to codes of length $N-1$, simply by not evaluating at the first element in the set $Z$ in Section \ref{sect1} (that is, at $\alpha_1=0$ or by not considering the coset $\mathfrak{I}_0$).

In addition, we emphasize that Theorem \ref{26} determines stabilizer quantum codes with designed distance, a lower bound for the dimension is also given. In a large number of cases, the dimension of our codes is strictly larger than the bound given in Theorem \ref{26}. Note that, in contrast with the minimum distance, the computation of the dimension of a linear code is not computationally intense and can be easily performed.

In the first two examples, we will detail the different values of $p, q, n$ and the considered length. However, for the sake or brevity and since it is straigthforward to deduce them from the paramenters of the codes, we do not give further details in the remaining examples. In Example \ref{ex:uno}, we obtain codes, both classical and quantum, that are records in \cite{codet}. For the rest of the examples there is no table of codes available (the previous table only contains binary stabilizer codes) and we indicate which codes exceed the quantum Gilbert-Varshamov bounds (QGVB, for short) \cite{mat, feng, kkk}.

\begin{exa}\label{ex:uno}
We consider the same setting as in examples \ref{ex:t1}, \ref{ex:t2} and \ref{ex:t3}. Let $p=2$, $s=1$, $n=4$. We obtain codes with length $q^{2n-1} = 2^7 =128$ over $q^{2s} = 4$. As a consequence, we are able to get 50 linear codes over $\mathbb{F}_4$ improving the parameters in \cite{codet}. In fact, we obtain two linear codes with parameters $[128,79,20]_4$ and $[128,75,22]_4$ improving the previous best known linear codes $[128,79,19]_4$ and $[128,75,21]_4$. We are also able to construct a $[128,85,16]_4$ code (no construction was known for such parameters in \cite{codet}). Then, by shortening the above codes, we obtain 50 linear codes over $ \mathbb{F}_4$ which are records at \cite{codet}. Their parameters can be found in Table \ref{TTabla1}. For the sake of brevity we only display some of them because their parameters are clear from their construction.

\begin{table}[ht]
%\caption{Nonlinear Model Results} % title of Table
\centering
%\begin{center}
\begin{tabular}{||c|c|c||c|c|c||c|c|c||c|c|c||}
  \hline \hline
 $n$ & $k$ & $d$  & $n$ & $k$ & $d$& $n$ & $k$ & $d$  & $n$ & $k$ & $d$ \\
  \hline \hline
 127& 84& 16& 126& 83& 16&125& 82& 16&124& 81& 16 \\
 123& 80& 16& 122& 79& 16&127& 78& 20&126& 77& 20 \\
 125& 76& 20& 124& 75& 20&123& 74& 20&122& 73& 20 \\
 121& 72& 20& 120& 71& 20&\ldots& \ldots& \ldots&105& 56& 20 \\
 127& 74& 22& 126& 73& 22& \ldots&  \ldots& \ldots &108& 55& 22 \\
\hline
 \hline
\end{tabular}
%\captionof{table}{Stabilizer affine variety ones codes over $\mathbb{F}_2$}
\caption{Linear codes over $\mathbb{F}_4$ which are records}
\label{TTabla1}
%\end{center}
\end{table}
These linear codes give rise to stabilizer quantum codes over $\mathbb{F}_2$, which by Theorem \ref{el1} are also records in the table \cite{codet}. We get stabilizer codes with parameters $[[128,80,10]]_2$ improving $[[128,80,9]]_2$; $[[128,72,11]]_2$ improving $[[128,72,10]]_2$; $[[128,66,12]]_2$ improving $[[128,66,11]]_2$ and $[[128,58,14]]_2$ improving $[[128,58,12]]_2$. Either puncturing or taking subcodes of the previous codes, we obtain binary stabilizer codes with  parameters as in Table \ref{TTabla2}.
\begin{table}[ht]
%\caption{Nonlinear Model Results} % title of Table
\centering
%\begin{center}
\begin{tabular}{||c|c|c||c|c|c||c|c|c||c|c|c||}
  \hline \hline
 $n$ & $k$ & $d$  & $n$ & $k$ & $d$& $n$ & $k$ & $d$  & $n$ & $k$ & $d$ \\
  \hline \hline
 128& 79& 10& 127& 80& 9&128& 71& 11&128& 65& 12 \\
 128& 64& 12& 128& 63& 12&128& 57& 14&128& 56& 14 \\
 128& 55& 14& 127& 58& 13&127& 57& 13&127& 56& 13 \\
\hline
 \hline
\end{tabular}
%\captionof{table}{Stabilizer affine variety ones codes over $\mathbb{F}_2$}
\caption{Quantum codes over $\mathbb{F}_2$ which are records}
\label{TTabla2}
%\end{center}
\end{table}
\end{exa}

\begin{exa}\label{ex:dos}
In this example, let $p=s=n=2$. We get stabilizer codes over $\mathbb{F}_4$. Some of these stabilizer codes with length $N = 64$, all of them with parameters that exceed the QGVB, are displayed in Table \ref{tabla1}.
\begin{table}[ht]
%\caption{Nonlinear Model Results} % title of Table
\centering
%\begin{center}
\begin{tabular}{||c|c|c||c|c|c||c|c|c||c|c|c||}
  \hline \hline
 $n$ & $k$ & $d$  & $n$ & $k$ & $d$& $n$ & $k$ & $d$  & $n$ & $k$ & $d$ \\
  \hline \hline
 64& 58& 3& 64& 54& 4&64& 50& 5&64& 48& 6 \\
 64& 44& 7&64& 40& 8& 64& 36& 9& 64& 34& 10 \\
 64& 30& 11&64& 26& 12&64& 22& 13&64& 20& 14\\
\hline
 \hline
\end{tabular}
%\captionof{table}{Stabilizer affine variety ones codes over $\mathbb{F}_2$}
\caption{Stabilizer codes over $\mathbb{F}_4$ of length $64$}
\label{tabla1}
%\end{center}
\end{table}

In the case where we do not evaluate at zero, their length is $63$ and we get stabilizer codes over  $\mathbb{F}_4$ with parameters as in Table \ref{tabla2}. Again, all the parameters of the presented codes exceed the QGVB.
\begin{table}[ht]
%\caption{Nonlinear Model Results} % title of Table
\centering
%\begin{center}
\begin{tabular}{||c|c|c||c|c|c||c|c|c||c|c|c||}
  \hline \hline
 $n$ & $k$ & $d$  & $n$ & $k$ & $d$& $n$ & $k$ & $d$  & $n$ & $k$ & $d$ \\
  \hline \hline
 63& 59& 3& 63& 55& 4&63& 51& 5&63& 49& 6 \\
 63& 45& 7&63& 41& 8& 63& 37& 9& 63& 35& 10 \\
 63& 31& 11&63& 27& 12&63& 23& 13&63& 21& 14\\
\hline
 \hline
\end{tabular}
%\captionof{table}{Stabilizer affine variety ones codes over $\mathbb{F}_2$}
\caption{Stabilizer codes over $\mathbb{F}_4$ of length $63$}
\label{tabla2}
%\end{center}
\end{table}

Notice that we get a large improvement with respect to the codes in \cite[Table III]{lag3}, and larger minimum distances ($10$ is the largest minimum distance in \cite[Table III]{lag3}).

We may consider quantum codes coming from complementary codes as well. Their length is $N^C= N^t - N = q^{2n} - N = 256 - 64 =192$. The parameters of some codes exceeding the QGVB are displayed in Table \ref{tabla3}. We have not found better codes over $\mathbb{F}_4$ with this length in the literature.
\begin{table}[ht]
%\caption{Nonlinear Model Results} % title of Table
\centering
%\begin{center}
\begin{tabular}{||c|c|c||c|c|c||c|c|c||c|c|c||}
  \hline \hline
 $n$ & $k$ & $d$  & $n$ & $k$ & $d$& $n$ & $k$ & $d$  & $n$ & $k$ & $d$ \\
  \hline \hline
192& 186& 3& 192& 182& 4&192& 178& 5&192& 174& 6 \\
 192& 170& 7&192& 166& 8& 192& 162& 9& 192& 158& 10 \\
 192& 154& 11&192& 150& 12&192& 146& 13&192& 21& 14\\
\hline
 \hline
\end{tabular}
%\captionof{table}{Stabilizer affine variety ones codes over $\mathbb{F}_2$}
\caption{Stabilizer codes over $\mathbb{F}_4$ of length $192$}
\label{tabla3}
%\end{center}
\end{table}

\end{exa}

\begin{exa}\label{ex:tres}
Table \ref{tabla4}  contains some stabilizer codes over $\mathbb{F}_3$ obtained with our procedure with length 242, 243 and 486. Our codes with length 242 and distance 5, 6, 10 and 11 exceed the the QGVB. Every code we give with length 243, but those with distance 15, 16 or 17, exceed the QGVB. Finally all codes with length 486 exceed that bound.
\begin{table}[ht]
%\caption{Nonlinear Model Results} % title of Table
\centering
%\begin{center}
\begin{tabular}{||c|c|c||c|c|c||c|c|c||c|c|c||}
  \hline \hline
 $n$ & $k$ & $d$  & $n$ & $k$ & $d$& $n$ & $k$ & $d$  & $n$ & $k$ & $d$ \\
  \hline \hline
242& 220& 5& 242& 214& 6&242& 208& 7&242& 202& 8 \\
 242& 196&10&242& 190& 11& 242& 184&12& 242& 178& 13 \\
 242& 172& 14&242& 166& 15&242& 160&16&242& 154& 17\\
 \hline
 243& 225& 5& 243& 219& 6&243& 213& 7&243& 207& 8 \\
 243& 201&9&243& 195& 11& 243& 189&12& 243& 183& 13 \\
 243& 177& 14&243& 171& 15&243& 165&16&243& 159& 17\\
\hline
486& 466& 5& 486& 460& 6&486& 454& 7&486& 448& 8 \\
 486& 442&9&486& 436& 11& 486& 430&12& 486& 424& 13 \\
 486& 418& 14&486& 412& 15&486& 406&16&486& 400& 17\\
 \hline
 \hline
\end{tabular}
%\captionof{table}{Stabilizer affine variety ones codes over $\mathbb{F}_2$}
\caption{Stabilizer codes over $\mathbb{F}_3$ of lengths $243$, $242$ and $486$}
\label{tabla4}
%\end{center}
\end{table}
\end{exa}

\begin{exa}\label{ex:cuatro}
Some stabilizer codes over $\mathbb{F}_5$ obtained with our procedure with length 124, 125 and 500 can be found in Table \ref{tabla5}. Our codes exceed the QGVB, excepting those with length 124 and distance 5 or 15. Notice that, again, we obtain a great improvement with respect to the codes with length 124 in \cite[Table III]{lag3}. In addition, the minimum distance of our codes can be much larger than in \cite{lag3}.
\begin{table}[ht]
%\caption{Nonlinear Model Results} % title of Table
\centering
%\begin{center}
\begin{tabular}{||c|c|c||c|c|c||c|c|c||c|c|c||}
  \hline \hline
 $n$ & $k$ & $d$  & $n$ & $k$ & $d$& $n$ & $k$ & $d$  & $n$ & $k$ & $d$ \\
  \hline \hline
124& 108& 5& 124& 106& 6&124& 102& 7&124& 98& 8 \\
 124& 94&9&124& 90& 10& 124& 88&11& 124& 84& 12 \\
 124& 80& 13&124& 76& 14&124& 72&15&124& 70& 16\\
 \hline
 125& 111& 5& 125& 107& 6&125& 105& 7&125& 101& 8 \\
 125& 97&9&125& 93& 10& 125& 89&11& 125& 87& 12 \\
 125& 83& 13&125& 79& 14&125& 75&15&125& 71& 16\\
\hline
500& 462& 11& 500& 458& 12&500& 454& 12&500& 450& 14 \\
 500& 446&15&500& 442& 16& 500& 438&17& 500& 434& 18 \\
 500& 430& 19&500& 426& 20&500& 422&21&500& 418& 22\\
 \hline
 \hline
\end{tabular}
%\captionof{table}{Stabilizer affine variety ones codes over $\mathbb{F}_2$}
\caption{Stabilizer codes over $\mathbb{F}_5$ of lengths $124$, $125$ and $500$}
\label{tabla5}
%\end{center}
\end{table}

\end{exa}

\begin{exa}\label{ex:cinco}
Finally, we display  Table \ref{tabla6}  containing stabilizer codes with length 342 and 2058 (from complementary codes) over $\mathbb{F}_7$. All the codes exceed the QGVB. Moreover, those with length 342 provide a great improvement with respect to the codes given in \cite[Table III]{lag3}. And as before, the minimum distance of our codes can be much larger than in \cite{lag3}.
\begin{table}[ht]
%\caption{Nonlinear Model Results} % title of Table
\centering
%\begin{center}
\begin{tabular}{||c|c|c||c|c|c||c|c|c||c|c|c||}
  \hline \hline
 $n$ & $k$ & $d$  & $n$ & $k$ & $d$& $n$ & $k$ & $d$  & $n$ & $k$ & $d$ \\
  \hline \hline
342& 326& 5& 342& 322& 6&342& 318& 7&342& 316& 8 \\
 342& 312&9&342& 308& 10& 342& 304&11& 342& 300& 12 \\
 342& 296& 13&342& 292& 14&342& 290&15&342& 286& 16\\
 342& 282& 17& 342& 278& 18&342& 274& 19&342& 270& 20 \\
 \hline
2058& 2020& 11& 2058& 2016& 12&2058& 2012& 12&2058& 2008& 14 \\
 2058& 2004&15&2058& 2000& 16& 2058& 1996&17& 2058& 1992& 18 \\
 2058& 1988& 19&2058& 1984& 20&2058& 1980&21&2058& 1976& 22\\
 2058& 1972& 23& 2058& 1968& 24&2058& 1964&25&2058& 1960& 26\\
 \hline
 \hline
\end{tabular}
%\captionof{table}{Stabilizer affine variety ones codes over $\mathbb{F}_2$}
\caption{Stabilizer codes over $\mathbb{F}_7$ of lengths $324$ and $2058$}
\label{tabla6}
%\end{center}
\end{table}
\end{exa}

\begin{rem}\label{rem:Grassl}
We have not performed an exhaustive search of good codes. We expect that more records can be found following this construction. For instance, Markus Grassl, with the setting as in Example \ref{ex:uno}, has found record complementary codes with the following parameters: $[127, 39, 44]_4$, $[127, 40, 43]_4$, $[127, 41, 42]_4$, $[128, 75, 22]_4$, $[128, 79, 20]_4$, $[128, 93, 14]_4$.

\end{rem}

\section*{Acknowledgment}

The authors thank Markus Grassl for pleasant discussions and for providing the codes in Remark \ref{rem:Grassl}.

\end{document}